\documentclass[submission,copyright,creativecommons]{eptcs}
\providecommand{\event}{Gandalf 2012} 
\usepackage{gandalf}

\title{Model-Checking Process Equivalences}
\author{Martin Lange \qquad Etienne Lozes \qquad\quad Manuel Vargas Guzm\'an  
\institute{School of Electrical Engineering and Computer Science \\
University of Kassel, Germany}
\email{\quad martin.lange@uni-kassel.de \qquad\qquad lozes@lsv.ens-cachan.fr
\qquad\qquad manuel.vargas@uni-kassel.de \quad}
}
\def\titlerunning{Model-Checking Process Equivalences}
\def\authorrunning{M.~Lange, E.~Lozes \& M.~Vargas~Guzm\'an}
\begin{document}
\maketitle

\begin{abstract}
Process equivalences are formal methods that relate programs and system which, informally, 
behave in the same way. Since there is no unique notion of what it means for two dynamic
systems to display the same behaviour there are a multitude of formal process equivalences,
ranging from bisimulation to trace equivalence, categorised in the linear-time branching-time
spectrum. 

We present a logical framework based on an expressive modal fixpoint logic which is capable
of defining many process equivalence relations: for each such equivalence there is a fixed
formula which is satisfied by a pair of processes if and only if they are equivalent with
respect to this relation. We explain how to do model checking, even symbolically, for a 
significant fragment of this logic that captures many process equivalences. This allows 
model checking technology to be used for process equivalence checking. We show how partial
evaluation can be used to obtain decision procedures for process equivalences from the generic
model checking scheme.
\end{abstract}

\section{Introduction}
In concurrency theory, a process equivalence is an equivalence relation between processes
--- represented as states of a labeled transition system (LTS) --- that aims at capturing the informal
notion of ``having the same behaviour''. A theory of behavioural equivalence obviously has applications
in formal systems design because it explains which programs or modules can be replaced by others without
changing the system's behaviour. 

There is no single mathematical notion of process equivalence as an equivalence relation on LTS. Instead
a multitude of different relations has been studied with respect to their pragmatics, axiomatisability,
computational complexity, etc.
These form a hierarchy with respect to containment, known as the 
\emph{linear-time branching-time spectrum} \cite{vG01}. We refer to the literature for a comprehensive overview over
all these equivalence relations at this point.

% An important aspect for any process equivalence notion is --- given the aforementioned application of process equivalences 
% in system design --- the question of how difficult it is to decide for two given processes whether or not they are equivalent
% with respect to this particular equivalence notion. It is known for example that deciding bisimilarity between finite-state
% processes can be done in polynomial time \cite{Alvarez:1991:PCD} whereas there are only exponential-time procedures for trace 
% equivalence so far because it is known to be PSPACE-hard \cite{IC::KanellakisS1990}. 

There are a few techniques which have proved to yield decision procedures for certain process equivalences, for example 
\emph{approximations} \cite{IC::KanellakisS1990}, \emph{characteristic formulas} \cite{IB-B941065,conf/icalp/CleavelandS91} 
and \emph{characteristic games} \cite{Stirling01,conf/cav/ShuklaHR96}. Often, for each equivalence notion, the same questions 
are being considered independently of each other, like ``can the algorithm be made to work with symbolic (BDD-based) 
representations of LTS?'', and the answer may depend on the technique being used to obtain the algorithm.

In this paper we introduce a further and generic, thus powerful technique, using the notion of \emph{defining formulas}. 
We present a modal fixpoint logic which is expressive enough to define these equivalences in the sense that, for an 
equivalence relation $R$, there is a fixed formula $\Phi_R$ which evaluates to true in a pair of processes if and only if 
they are related by $R$. We also give a model checking algorithm for this logic. This can then be instantiated with such 
formulas $\Phi_R$ in order to obtain an equivalence checking algorithm for $R$. Furthermore, the model checking algorithm 
can easily deal with symbolic representations. Thus, this yields BDD-based equivalence checking algorithms for all the process 
equivalences mentioned in this paper. Moreover, with this generic framework, the task of \emph{designing} an equivalence 
checking algorithm for any new equivalence notion boils down to simply \emph{defining} this relation in the modal fixpoint 
logic presented here.

This is related to work on \emph{characteristic formulas}, yet it is different. There, in order to
check two processes $P$ and $Q$ for, say, bisimilarity, one builds the characteristic formula
$\Phi_\sim^P$ describing all processes that are bisimilar to $P$ and checks whether or not $Q \models
\Phi_\sim^P$ holds. Here, we take a fixed formula $\Phi_\sim$ and check whether or not $(P,Q) \models
\Phi_\sim$ holds. Note that the former cannot be made to work with a symbolic representation of $P$
whereas the latter can. In general, using defining instead of characteristic formulas has the advantage
of lifting more model checking technology to process equivalence checking. 

% The latter has, so far,
% sometimes been seen as less developed than the former \ML{are there good citations} simply because of
% the lack of a unifying framework. There are connections between model checking and equivalence
% checking. For instance it is known that two processes are trace equivalent if and only if they satisfy
% exactly the same LTL properties. Similarly, two image-finite processes are bisimilar if and only if
% they satisfy the same properties described in any branching-time logic like CTL, CTL$^*$ or the modal
% $\mu$-calculus. Note, however, that such characterisations are not effective: model checking two
% processes against all formulas of a logic is clearly not feasible. For finite processes this
% characterisation can be weakened to all formulas with a bounded nesting depth of temporal operators
% only which makes it effective but hardly efficient.

The use of fixed formulas expressing process equivalences is being made possible by the design of a new
modal fixpoint logic. It is obtained as the merger between two extensions of the modal $\mu$-calculus, namely
the \emph{higher-dimensional $\mu$-calculus} $\mathcal{L}_\mu^\omega$ \cite{Otto99} and the
\emph{higher-order $\mu$-calculus} HFL \cite{ViswanathanV04}. The former allows formulas to make
assertions about tuples of states rather than states alone. This is clearly useful in this setup given
that process equivalences are binary relations. Not surprisingly, it is known for instance that there
is a formula in $\mathcal{L}^2_\mu$ --- the fragment speaking about tuples of length 2 ---
that expresses bisimilarity. On the other hand, HFL's higher-order features allow the logic to express
properties that are more difficult than being polynomial-time decidable. It is known for instance that it
can make assertions of the kind ``for every finite word $w$ there is a path labeled with $w$'' which is
very useful for describing variants of trace equivalence.

The rest of the paper is organised as follows. Sect.~\ref{sec:proceq} recalls the linear-time
branching-time hierarchy. For the sake of completeness, the exact definitions of these relations are
presented in an appendix. Sect.~\ref{sec:hohd} defines the aforementioned modal fixpoint logic.
Sect.~\ref{sec:procequivformula} realises the reduction from process equivalence checking to model
checking fixed formulas by simply spelling out the definition of those equivalence notions in this
modal fixpoint logic. Sect.~\ref{sec:mc} shows how to do model checking for the fragment of
this logic which is most significant to process equivalence checking, and how the na\"{\i}ve model
checking algorithm can be optimised using need-driven function evaluation and partial evaluation.
Sect.~\ref{sec:concl} concludes with ideas on further work in this direction.

%%% Local Variables: 
%%% mode: latex
%%% TeX-master: "main"
%%% End: 

\section{Process Equivalences}
\label{sec:proceq}
In this section we present the hierarchy of the linear-time branching-time spectrum, as it
can be seen from Fig.~\ref{fig:ltbth}, the greatest equivalence is \emph{finite trace 
equivalence}, and the finest one is \emph{bisimulation}. First we introduce some 
preliminaries and notation.
We use letters $a,b,\ldots$ to denote actions, and letter $t$ to denote a trace. Letters $P,Q,\ldots$ denote processes.

A \textit{labeled transition system} (LTS) over a set of actions\footnote{For
  simplicity, we do not consider state labels.} $\Act = \{a,b,...\}$  is a triple $(\Procs,\Act,\tr{})$, where $\Procs$ is a
set of states representing processes, $\Act$ is the set of actions, and $\tr{} \; \subseteq \Procs
\times \Act \times \Procs$ is a transition relation.  We write $P \tr{a} Q$ for $(P,a,Q) \in \tr{}$.
$I(P) := \{a \in \Act \mid \exists Q. P \tr{a} Q \}$ denotes the set of \emph{initial actions} of a
process $P$.

A \emph{finite trace} $\:t \in \Act^*$ of $P_0$, is a finite sequence of actions $a_1 ... a_n$ s.t. there are $P_0 ... P_n$ with 
$P_{i-1} \tr{a_i} P_n$ for all $i = 1, ..., n$. We write $P \tr{t} Q$ if there is a trace $t$ of $P$ that ends
in $Q$.

Since the main purpose of this paper is not to focus deeply on the semantics of process equivalences, we do not address 
the definitions in this section. For further details, the reader can find the exact definitions of all process equivalences in Appendix ~\ref{sec:app}.

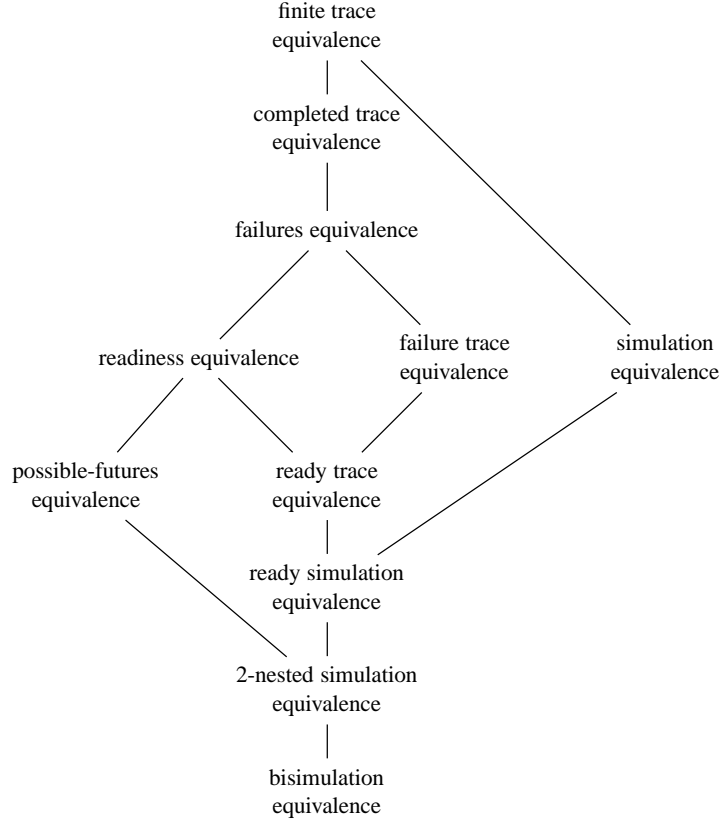
\begin{figure}[t]
 \begin{center}
  \scalebox{0.8}{\begin{tikzpicture}[shorten >=1pt,auto,text width=9em,text badly centered,node distance=3cm,semithick]

  \node (A)                                  {finite trace equivalence};
  \node (B) [node distance=1.7cm,below of=A] {completed trace equivalence};
  \node (C) [node distance=1.7cm,below of=B] {failures equivalence};
  \node (D) [below left of=C]                {readiness equivalence};
  \node (E) [below right of=C]               {failure trace equivalence};
  \node (F) [below left of=E]                {ready trace equivalence};
  \node (G) [node distance=1.7cm,below of=F] {ready simulation equivalence};
  \node (H) [node distance=1.7cm,below of=G] {2-nested simulation equivalence};
  \node (I) [node distance=1.7cm,below of=H] {bisimulation equivalence};
  \node (J) [node distance=3.5cm,right of=E] {simulation equivalence};
  \node (K) [node distance=4cm,left of=F]    {possible-futures equivalence};
  
  \path (A) edge node {} (B)
            edge node {} (J)
        (B) edge node {} (C)
        (C) edge node {} (D)
        (C) edge node {} (E)
        (D) edge node {} (F)
            edge node {} (K)
        (E) edge node {} (F)
        (F) edge node {} (G)
        (G) edge node {} (H)
        (H) edge node {} (I)
        (J) edge node {} (G)
        (K) edge node {} (H);
  \end{tikzpicture}}
 \end{center}
\caption{\label{fig:ltbth}The linear-time branching-time hierarchy.}
\end{figure}

\section{A Higher-Order Higher-Dimensional $\mu$-Calculus}
\label{sec:hohd}
\subsection{Combining Higher-Order and Higher-Dimensionality}

In this section, we introduce a logical formalism, called $\HOHDMU$ that extends the standard 
modal $\mu$-calculus. It can be seen as the combination of two 
extensions of the $\mu$-calculus that were 
previously defined: the higher-order fixpoint logic HFL~\cite{ViswanathanV04}, 
and the higher-dimensional modal mu-calculus 
$\mathcal{L}_\mu^\omega$~\cite{Otto99}. First we build some intuition about the use of higher-order
and higher-dimensional features in modal logics.

In HFL, formulas may denote
not only sets of processes, but also \emph{predicate transformers}, 
\emph{i.e.} functions from sets of processes to
sets of processes, and more generally any higher-order functions of some
functional type built on top
of the basic type $\PR$ of set of processes.
For instance, the formula 
$$
\lambda x:\PR.\ \may{a}{}x \;\wedge\; \must{b}{}\bot
$$ 
denotes the function that 
takes a predicate $\form$ of type $\PR$, \emph{i.e.} a set of processes, 
and returns the predicate $\may{a}{}\form \;\wedge\; \must{b}{}\bot$,
\emph{i.e} the set of processes 
$P$ for which $P\stackrel{b}{\not\rightarrow}$ and $P\tr{a}P'$ for some $P'\models \form$.
Similarly, the formula 
$$
\lambda f: \PR\to\PR.\ \lambda x:\PR.\ f\ (f\ x)
$$
denotes the function that maps any predicate transformer $f$ 
to the predicate transformer $f^2$.

Like in the standard $\mu$-calculus, to every monotone function of type $\PR\to\PR$ denoted by a
formula $\lambda x:\PR.\ \form$, HFL associates a least fixed point $\mu x:\PR.\ \form$. In HFL, this
construction generalises well to any monotone function of type $\typ\to\typ$, thanks to a construction
based on the pointwise ordering of functions we recall below. For instance, the formula $\mu
f:\PR\to\PR.\ f$ denotes the constant function $\lambda x.\;\bot$, since it is the smallest predicate
transformer, according to the pointwise ordering, that is fixed by the identity function.  A bit more
elaborated, the formula
$$
\mu f:\PR\to\PR.\ \lambda x:\PR.\lambda y:\PR.\;\; (x\wedge y)\;\; \vee\;\; f\ \may{a}{}x\ \may{b}{}y
$$
can be unfolded as 
$$
f\ x\ y\quad=\quad (x\wedge y)\;\;\vee\;\; f\ \may{a}{}x\ \may{b}{}y \quad = \quad
(x\wedge y)\;\;\vee\;\; (\may{a}{}x\wedge \may{b}{}y)\;\;\vee\;\; f\ \may{aa}{}x\ \may{bb}{}y\quad =\quad \dots
$$
and thus denotes the function  
$\lambda x,y.\ \bigvee_{n\geq 0}\;\may{a}{}^nx\;\wedge\;\may{b}{}^ny$.

The higher-dimensional $\mu$-calculus extends the $\mu$-calculus in a different way. In $\mathcal{L}_{\mu}^{\omega}$, 
logical formulas do not denote sets of processes,
but sets of tuples of processes. The $i$-th component of a tuple
can be changed by the $i$-th modality $\may{a}{i}$.
For instance, the 2-dimensional formula $\may{a}{1}\top\;\wedge\;\may{b}{2}\top$ 
denotes the set of
pairs $(P,Q)$ such that $P\tr{a}P'$ and $Q\tr{b} Q'$ for some $P',Q'$. 
The modality $\may{a}{i}$ only modifies the $i$-th component of the tuple,
and leaves all other components unchanged, which validates some rules like
$$
\begin{array}{l@{\qquad\qquad}r@{\quad\Leftrightarrow\quad}l@{\qquad}l}
\mbox{(commutation)}& \may{a}{1}\may{b}{2}\;\form & \may{b}{2}\may{a}{1}\;\form
\\
\mbox{(scope extrusion)} & \may{a}{d}\;(\form\;\wedge\;\formbis) & \may{a}{d}\;\form\;\;\wedge\;\; \formbis & \big(\mathsf{dim}(\formbis)<d\big).
\end{array}
$$
%% Note however that $\may{a}{1}\must{a}{2}\;\form$ is not the same
%% as $\must{a}{2}\may{a}{1}\;\form$. Assume for instance 
%% that $\form$ characterizes pairs $(P,Q)$ such that $P\not\sim_S Q$. 
%% Then $(P,Q)\;\models\;\may{a}{1}\must{a}{2}\;\form$ 
%% implies $P\not\sim_S Q$ --- and moreover the absence of simulation equivalence
%% can be observed by first playing an $a$ transition on $P$'s side
%% in the simulation game --- whereas $(P,Q)\;\models\; \must{a}{2}\may{a}{1}\;\form$ does not necessarily imply $P\not\sim Q$.

We associate a type
$\iPR{d}$ to the formulas of the $d$-dimensional $\mu$-calculus. 
Note that there is a significant difference
between \emph{e.g.} $\iPR{2}$ and the product type 
$\PR{}\times\PR$: the former is the type of binary predicates over processes,
whereas the latter is the type of pairs of unary predicates. 
There is indeed no obvious way of representing $\mathcal{L}_{\mu}^{\omega}$ in
HFL, although HFL may encode some of product types using standard techniques.

%% For instance, the formula $\lambda x:\PR.\lambda y:\PR. \may{a}{}x \wedge \may{b}{}y$
%% denotes a function that takes a set of processes and returns a function that takes
%% a set of processes, which, due to the type isomorphy $\PR \to (\PR \to \PR) \simeq \PR \times \PR \to \PR$, 
%% can equally be seen as a function that takes two sets of processes and returns one such set. 
%% However, it is impossible to define functions that return such tuples. Moreover, 
%% tuples of sets are weaker than sets of tuples. It is impossible to define a function in HFL that
%% takes a set of pairs of processes. 
%% This feature is enabled by higher-dimensionality in $\HOHDMU$. 

\subsection{Syntax and Semantics}

Let $\Act$ be as above. Fix $d \in \mathbb{N}$.
We assume an infinite set $\Var = \{x,y,z,\ldots\}$ of variables. 
A formula is a 
$\form$ that 
can be derived from by 
$$
\begin{array}{rl@{\qquad}r}
\form,\formbis \enspace ::= 
& \top\mid\may{a}{i}\form\mid\neg\form\mid\form\wedge\formbis\mid
x\mid\lambda x^{v}:\typ.\;\form \mid \mu x:\typ.\form \mid \form\ \formbis
&
\mbox{(formulas)}
\\[.8em]
v \enspace ::= &
\mon\mid\comon\mid\bimon 
& \mbox{(variances)}
\\[.8em]
\typ,\typbis \enspace ::= 
& \iPR{d}\mid\typ^v\to\typbis
& \mbox{(types)}
%\\[.8em]
% \Gamma ::= 
% & \emptyset\mid\Gamma\; ,\; x^v:\typ
% & \mbox{(environment)}
\end{array}
$$
where $1 \le i \le d$, $a \in \Act$ and $x \in \Var$.

The typing arrow $\to$ is --- as usual --- right-associative. Thus, every type is of the form 
$\typ = \typ_1^{v_1} \to \ldots \to \typ_m^{v_m} \to \iPR{d}$ for some $m \ge 0$. For such normalised
types we can define their \emph{order} simply as $\ord{\typ} := \max \{ 1 + \ord{\typ_i}: i=1,\ldots,m \}$ with the convention
of $\max \emptyset = 0$.

Formulas are ruled by the type system depicted on Fig.~\ref{fig:type-system}.
Intuitively, the aim of the type system is to prevent applications of non-functions
to formulas, as well as fixpoint definitions of non-monotone functions, like $\mu x.\lambda y. \neg x\ y$.
In order to exclude the latter, \emph{variances} are introduced for each function parameter.

For $d \ge 1$ and $o \ge 0$, let $\hohdmu{o}{d}$ consist of all closed formulas $\Phi$ such that the statement 
$\emptyset \vdash \Phi: \iPR{d}$ is typable, and each type annotation in $\Phi$ has order at most $o$.
In general, a statement of the form $\Gamma \vdash \Phi : \typ$ 
asserts that the formula $\Phi$ has type $\typ$ under the assumptions $\Gamma$, which is a list of the
form $x_1^{v_1}: \typ_1, \ldots, x_m^{v_m} : \typ_m$. For such a list of assumptions, $\neg\Gamma$ is 
obtained from $\Gamma$ by swapping the variance of each variable: $\mon$ becomes $\comon$ and vice-versa, 
and $\bimon$ remain the same. 
Thus, $\hohdmu{o}{d}$ consists of all well-typed 
and closed formulas of type that should denote a set of $i$-tuples in an LTS and use at most higher-order
features of order $o$. Let 
\begin{displaymath}
\hohdmu{o}{\omega} \enspace := \enspace \bigcup_{d \ge 1} \hohdmu{o}{d} \quad, \qquad 
\hohdmu{\omega}{d} \enspace := \enspace \bigcup_{o \ge 0} \hohdmu{o}{d} \quad, \qquad
\HOHDMU \enspace := \enspace \bigcup_{o \ge 0}\bigcup_{d \ge 1} \hohdmu{o}{d}
\end{displaymath}

\begin{figure}
\begin{mathpar}
\trule{}{\judg{\Gamma}{\top:\iPR{d}}}

\trule{\judg{\Gamma}{\form:\iPR{d}}\\i\leq d}{\judg{\Gamma}{\may{a}{i}\form:\iPR{d}}}

\trule{\judg{{}\neg(\Gamma)}{\form:\iPR{d}}}{\judg{\Gamma}{\neg\form:\iPR{d}}}

\trule{\judg{\Gamma}{\form:\iPR{d}}\\\judg{\Gamma}{\formbis:\iPR{d}}}
{\judg{\Gamma}{\form\wedge\formbis:\iPR{d}}}

\trule{v\in\{\mon,\bimon\}}{\judg{\Gamma\;,\;x^v:\typ}{x:\typ}}

\trule{\judg{\Gamma,x^v:\typbis}{\form:\typ}}
{\judg{\Gamma}{\lambda x^v:\typbis.\ \form :\typbis^v\to\typ}}

\trule{\judg{\Gamma,x^{\mon}:\typ}{\form:\typ}}
{\judg{\Gamma}{\mu x(y_1,\ldots,y_m):\typ.\ \form :\typ}}

\trule{\judg{\Gamma}{\form:\typbis^{\mon}\to\typ}\\\judg{\Gamma}{\formbis:\typbis}}
{\judg{\Gamma}{\form\ \formbis :\typ}}

\trule{\judg{\Gamma}{\form:\typbis^{\comon}\to\typ}\\\judg{\neg(\Gamma)}{\formbis:\typbis}}
{\judg{\Gamma}{\form\ \formbis :\typ}}

\trule{\judg{\Gamma}{\form:\typbis^{\bimon}\to\typ}\\
\judg{\Gamma}{\formbis:\typbis}\\\judg{\neg(\Gamma)}{\formbis:\typbis}
}
{\judg{\Gamma}{\form\ \formbis :\typ}}
\end{mathpar}
%with $\neg(\mon)=\comon$, $\neg(\comon)=\mon$, $\neg(\bimon)=\bimon$, 
%$\neg(\nomon)=\nomon$, and $\neg(x_1^{v_1}:\typ_1,\dots,x_n^{v_n}:\typ_n)~\eqdef~x_1^{\neg v_1}:\typ_1,\dots,x_n^{\neg v_n}:\typ_n$. 
%% $\neg v=\left\{\begin{array}{l@{\quad\mbox{if}~v=~}l}
%% \comon & \mon \\ \mon & \comon \\ \bimon & \bimon \\ \nomon & \nomon\end{arra%% y}\right.$.

\caption{\label{fig:type-system}The type system of \HOHDMU.}
\end{figure}

Before we can explain the semantics of a formula we need to give the types a semantics too. Let
$\Transsys_i = (\Procs_i,\Act,\tr{}_i)$ for $i=1,\ldots,d$ be LTS. We take them to be fixed and 
simply write $\sem{\tau}$ instead of $\sem{\tau}^{\Transsys_1,\ldots,\Transsys_d}$. The semantics 
of a type is inductively defined as follows. 
\begin{itemize}
\item $\sem{\iPR{d}}$ is the set of all sets of $d$-ary predicates of processes, ordered by inclusion,
  i.e.\ $\sem{\iPR{d}}=(\pset{\Procs_1\times\ldots\times\Procs_d},\leq_{\iPR{d}})$, with $\mathcal S\leq_{\iPR{d}} \mathcal S'$
  if $\mathcal S\subseteq \mathcal S'$.
\item $\sem{\typ^{\mon}\to\typbis}$ is the set of monotone functions from $\sem{\typ}$ to
  $\sem{\typbis}$, ordered by pointwise ordering, i.e.\  $\sem{\typ^{\mon}\to\typbis}=\{f \in
  \sem{\typbis}^{\sem{\typ}}\;:\;\forall x,y.\;x\leq_{\typ}y \Rightarrow f(x)\leq_{\typbis} f(y)\;\}$
  and $f\leq_{\typ^{\mon}\to\typbis} g$ if for all $x\in\sem{\typ}$, $f(x)\leq_{\typbis} g(x)$.
\item similarly $\sem{\typ^{\comon}\to\typbis}$ is the set of co-monotone functions, and
  $\sem{\typ^{\bimon}\to\typbis}$ is the set of all functions
  from $\sem{\typ}$ to $\sem{\typbis}$, ordered by pointwise ordering.
\end{itemize}
All these domains are complete lattice. As a consequence, any function 
$f\in\sem{\typ^{\mon}\to\typ}$ has a least fixpoint according to the Knaster-Tarski
Theorem~\cite{Kna28,Tars55}; we write 
$\LFP{\typ} \; f$ to denote it.

The semantics of a formula $\form$ of type $\typ$ with respect to an environment $\Gamma$, the
underlying LTS $\Transsys_1,\ldots,\Transsys_d$ and an
interpretation $\eta$ of its free variable is an element of $\sem{\typ}$, defined as follows.
Let $\Procs^d := \Procs_1 \times \ldots \times \Procs_d$.
\begin{align*}
\sem{\judg{\Gamma}{\top:\iPR{d}}}_{\eta} \enspace &= \enspace \Procs^d \\
\sem{\judg{\Gamma}{\may{a}{i}\form:\iPR{d}}}_{\eta} \enspace &= \enspace 
\{ (P_1,\ldots,P_d) \in\Procs^d\;:\;\exists P_i'\in\Procs_i.\;P_i\tr{a}P_i'~\mbox{and}~\\
&\qquad\qquad(P_1,\dots,P'_i,\ldots P_d)\in\sem{\judg{\Gamma}{\form:\iPR{d}}}_{\eta} \} \\
\sem{\judg{\Gamma}{\neg\form:\iPR{d}}}_{\eta} \enspace &= \enspace \Procs^d \setminus \sem{\judg{\neg(\Gamma)}{\form:\iPR{d}}}_{\eta} \\
\sem{\judg{\Gamma}{\form\wedge\formbis:\iPR{d}}}_{\eta} \enspace &= \enspace 
\sem{\judg{\Gamma}{\form:\iPR{d}}}_{\eta}~\cap~
\sem{\judg{\Gamma}{\formbis:\iPR{d}}}_{\eta}  \\
\sem{\judg{\Gamma,x^v}{x^v}}_{\eta} \enspace &= \enspace \eta(x) \\
\sem{\judg{\Gamma}{\lambda x^v:\typbis.\;\form\;:\;\typ}}_{\eta} \enspace &= \enspace
f~\mbox{such that for all}~e\in\sem{\typbis},\quad f(e)~=~
\sem{\judg{\Gamma,x^v:\typbis}{\form:\typ}}_{\eta[x\mapsto e]} \\
\sem{\judg{\Gamma}{\mu x:\typ.\;\form\;:\;\typ}}_{\eta} \enspace &= \enspace
\LFP{\typ}~\sem{\judg{\Gamma}{\lambda x^{\mon}:\typ.\;\form\;:\;\typ}}_{\eta} \\
\sem{\judg{\Gamma}{\form\; \formbis:\typ}}_{\eta} \enspace &= \enspace 
f(e),\mbox{ where }
f=\sem{\judg{\Gamma}{\form:\typbis^v\to\typ}}_{\eta} 
\mbox{ and } 
e=\sem{\judg{\Gamma'}{\formbis:\typbis}}_{\eta}
 \\
\end{align*}
If $\judg{}{\form:\typ}$ and $m\in\sem{\typ}$, we write $m\models\form$ to denote that
$m\in\sem{\judg{}{\form:\typ}}$.

We assume standard notations for derived boolean and modal operators, and write $\form\vee\formbis$ for
$\neg(\neg\form\wedge\neg\formbis)$, or $\must{a}{i}\form$ for $\neg\may{a}{i}\neg\form$, or
$\form\nLeftrightarrow\formbis$ for $(\form\wedge\neg\formbis)\vee(\neg\form\wedge\formbis)$, etc.  If
$\judg{\Gamma}{\form:\typ_1^{\mon}\to\typ_2}$ and $\judg{\Gamma}{\formbis:\typ_2^{\mon}\to\typ_3}$ are
two monotone functions, we write $\formbis\circ\form$ as a shorthand for the monotone function $\lambda
x^{\mon}:\typ_1.\;\formbis\;(\form\; x)$. We will also write $\mu x(y_1,\ldots,y_m):\sigma_1^{v_1} \to
\ldots \to \sigma_m^{v_m} \to \tau.\Phi$ instead of $\mu x:\tau.\lambda y_1:\sigma_1^{v_1} \ldots
\lambda y_m:\sigma_m^{v_m}.\Phi$. Finally, $\Phi[\Psi/x]$ is obtained from $\Phi$ by replacing every
free occurrence of the variable $x$ with the formula $\Psi$.

%%% Local Variables: 
%%% mode: latex
%%% TeX-master: "main"
%%% End: 

\section{Process Equivalences as Formulas}
\label{sec:procequivformula}
In this section, we show how all process equivalences
of the linear-time branching-time hierarchy can be characterised
by $\HOHDMU$ in a certain sense. To improve readability, we will often keep the type system implicit, 
and use different variable symbols in order to suggest the type. For instance,
we write $X,Y$ to range over sets of tuples of 
processes, $F,G$ to range over first-order
functions of type 
$\iPR{2}^{v_1}\to\dots\to\iPR{2}^{v_m}\to\iPR{2}$, 
whereas $\mathcal F$ ranges over second-order functions.
We write $\swap{\form}$ for the formula $\form$ in
which $\may{a}{1}$ and $\may{a}{2}$ are swapped for any $a \in \Act$, equally for $\must{a}{1}$ and $\must{a}{2}$. For any $t = a_1\ldots a_n \in \Act^*$
we write $\may{t}{i}\Phi$ to abbreviate $\may{a_1}{i}\ldots\may{a_n}{i}\Phi$, and similarly for $\must{t}{i}$.

We say that an equivalence relation 
$\mathcal R$ over processes is \emph{characterised} by a closed 
formula $\form$ of type $\iPR{2}$ if for all processes $P,Q$
$$
P\mathrel{\mathcal R}Q \qquad \Leftrightarrow \qquad (P,Q)\models \form.
$$
We will say that a formula 
$\form$ \emph{tests} for $\mathcal R$ if $\neg\form\wedge\neg\swap{\form}$
characterises $\mathcal R$. Intuitively,
$\form$ tests for $P\mathrel{\mathcal R}Q$ if it is true when
$P$ presents a behavior that $Q$ cannot reproduce. 
For readability, we only present formulas
that test process equivalences, but it is straightforward to 
derive formulas that characterise process equivalence.
We later write $\form_{\mathcal R}$ for a formula that tests $\mathcal R$.

Let us first consider trace equivalence. If we
were to consider a logic with infinite disjunctions,
a formula testing finite trace equivalence would be
$\bigvee_{t\in\Act^*}\may{t}{1}\top\wedge\must{t}{2}\bot$.
Encoding such an infinite disjunction is not easy in general,
and it is indeed impossible in the ordinary $\mu$-calculus.
But the $\HOHDMU$ formula
$$
\form_{\mathsf{t}} \qquad \eqdef \qquad \big(\mu F(X,Y). \quad (X\wedge Y)\quad \vee \bigvee_{a\in \Act} F \quad \may{a}{1}X \quad \must{a}{2}Y\big)\qquad \top \qquad \bot
$$
is equivalent to the one with the infinite disjunction, and thus
tests trace equivalence.

Let us consider now all other equivalences of the lower part of the hierarchy.
As all 
these equivalences are variations around finite trace equivalence,
it can be expected that the formulas testing them are very similar.
We introduce the template formula $\mathsf{TemplateTrace}(\mathsf{Mod},\mathsf{Pred})\eqdef$
$$
\bigvee_{\form\in \mathsf{Pred}}
\Big(\mu F(X,Y).\quad (X\wedge Y)\quad \vee \bigvee_{\Psi\in \mathsf{Mod}} F\quad \Psi \quad \neg \swap{\Psi}[\neg Y/X]\Big)\qquad \form \qquad \neg \form[1\leftrightarrow 2]
$$
for some finite sets $\mathsf{Pred}$ and $\mathsf{Mod}$ of 0-order 
formulas. For instance, the above formula testing
trace equivalence is obtained for $\mathsf{Pred}=\{\top\}$ and
$\mathsf{Mod}=\{\may{a}{1}X \,:\, a\in \Act\}$. Other
instantiations of these two parameters provide all equivalences above
simulations, c.f.\ the upper table in Fig.~\ref{fig:belowsim}. Let 
$\mathsf{fail(A)}\eqdef \bigwedge_{a\in A}\must{a}{1}\bot$ and
$\mathsf{ready}(A)\eqdef \bigwedge_{a\in A}\may{a}{1}\bot\,\wedge\,\bigwedge_{a\not\in A}\must{a}{1}\bot$.

\begin{figure}
\begin{center}
\begin{tabular}{@{\quad}l@{\quad}|@{\quad}c@{\quad}|@{\quad}c@{\quad}}
%\hline 
\rule[-2mm]{0pt}{0mm}equivalence  & $\mathsf{Mod}$ & $\mathsf{Pred}$ \\
\hline\hline
\rule[-2mm]{0pt}{6mm}trace & $\{\,\may{a}{1}X \,:\, a\in \Act\,\}$ 
& $\{\,\top\,\}$
\\ \hline 
\rule[-2mm]{0pt}{6mm}completed trace 
& $\{\,\may{a}{1}X \,:\, a\in \Act\,\}$
& $\{\,\bigwedge_{a\in\Act}\must{a}{1}\bot\,\}$
\\ \hline 
\rule[-2mm]{0pt}{6mm}failure
& $\{\,\may{a}{1}X \,:\, a\in \Act\,\}$
& $\{\,\mathsf{fail}(A)\,:\,A\subseteq \Act\}$
\\ \hline 
\rule[-2mm]{0pt}{6mm}failure trace
& $\{\,\may{a}{1}X \,:\, a\in \Act\,\}\, \cup\, \{\,X\wedge\mathsf{fail}(A):\,A\subseteq \Act\,\}$
& $\{\top\,\}$ 
\\ \hline
\rule[-2mm]{0pt}{6mm}readiness
& $\{\,\may{a}{1}X \,:\, a\in \Act\,\}$
& $\{\,\mathsf{ready}(A)\,:\,A\subseteq \Act\}$
\\ \hline 
\rule[-2mm]{0pt}{6mm}ready trace
& $\{\,\may{a}{1}X \,:\, a\in \Act\,\}\, \cup\, \{\,X\wedge\mathsf{ready}(A) \,:\,A\subseteq \Act\,\}$
& $\{\,\top\,\}$
%\\ \hline
\end{tabular}
\vskip4mm

\begin{tabular}{@{\quad}l@{\quad}|@{\quad}c@{\quad}|@{\quad}c@{\quad}}
%\hline 
\rule[-2mm]{0pt}{0mm}equivalence  & $\mathsf{Mod}$ & $\mathsf{Test}$ \\
\hline\hline
\rule[-2mm]{0pt}{6mm}simulation 
& $\{\,\may{a}{1}\must{a}{2}X \,:\, a\in \Act\,\}$ 
& $\bot$
\\ \hline 
\rule[-2mm]{0pt}{6mm}completed simulation 
& $\{\,\may{a}{1}\must{a}{2}X \,:\, a\in \Act\,\}$
& $\mathsf{deadlock}_1\nLeftrightarrow\mathsf{deadlock}_2$
\\ \hline 
\rule[-2mm]{0pt}{6mm}ready simulation
& $\{\,\may{a}{1}\must{a}{2}X \,:\, a\in \Act\,\}$
& $\bigvee_{A\subseteq \Act}\mathsf{ready}_1(A)\nLeftrightarrow\mathsf{ready}_2(A)$
\\ \hline 
\rule[-2mm]{0pt}{6mm}2-nested simulation
& $\{\,\may{a}{1}\must{a}{2}X \,:\, a\in \Act\,\}$
& $\swap{\Phi_{\mathsf{s}}}$  
\\ \hline
\rule[-2mm]{0pt}{6mm}bisimulation
& $\{\,\may{a}{1}\must{a}{2}X\, , \, \may{a}{2}\must{a}{1}X \,:\, a\in \Act\,\}$
& $\bot$
% \\ \hline 
% %\multirow{2}{*}{\mbox{ready trace}}
% %& \{\,\may{a}{1}X \,:\, a\in \Act\,\}\, \cup\,
% %& \multirow{2}{*}{$\{\,\top\,\}$}
% \rule[-2mm]{0pt}{6mm}\mbox{ready trace}
% & \{\,\may{a}{1}X \,:\, a\in \Act\,\}\, \cup\, \{\,X\wedge\mathsf{ready}(A) \,:\,A\subseteq \Act\,\}
% & \{\,\top\,\} \\[-4mm]
% %\\ \hline
\end{tabular}
\end{center}
\caption{\label{fig:belowsim}Instantiations of the parameters for the template formulas.}
\end{figure}

Formulas testing the relations below simulation equivalence can also be derived from a common, but simpler template. 
In these case, no higher-order features are needed. Let $\mathsf{TemplateSim}(\mathsf{Mod},\mathsf{Test}) \eqdef$
\begin{displaymath}
\mu X.\enspace \mathsf{Test}\enspace \vee \enspace\bigvee\limits_{\Psi \in \mathsf{Mod}} \Psi
\end{displaymath} 
where $\mathsf{Test}$ stands for an $\hohdmu{0}{2}$ formula, and $\mathsf{Mod}$ is a finite
set of $\hohdmu{0}{2}$ formulas. The instantiations for the respective equivalence relations are presented in the
lower table of Fig.~\ref{fig:belowsim}. In the case of 2-nested simulation equivalence, $\Phi_{\mathsf{s}}$
stands for the formula that is obtained from this template for simulation equivalence. We define
$\mathsf{deadlock}_i \eqdef \bigwedge_{a \in \Act} \must{a}{i}\bot$. 

The only equivalence that is shown in Fig.~\ref{fig:ltbth} but not dealt with so far is possible-futures
equivalence. It is definable in $\hohdmu{2}{2}$ through
$$
\Big(\mu \mathcal F. \quad \lambda G_1,G_2. \lambda X. \quad G_1\ (G_2\ X) \,\vee\, \bigvee_{a\in\Act} 
\big(\mathcal F \ (\may{a}{1}\circ G_1) \ (\must{a}{2}\circ G_2) \ X\big)\Big)\quad \lambda X.X\quad \lambda X.X \quad \formbis_{\mathsf{t}}
$$
where $\formbis_{\mathsf{t}}=\form_{\mathsf{t}}\vee\swap{\form_{\mathsf{t}}}$ is the negation of the characteristic formula for trace equivalence.
It remains to be seen whether or not it is also definable in $\hohdmu{1}{2}$ like the other equivalences
are.

%%% Local Variables: 
%%% mode: latex
%%% TeX-master: "main"
%%% End: 

\section{Model-Checking $\hohdmu{1}{2}$}
\label{sec:mc}
\subsection{From Model Checking to Process Equivalence Checking}

The characterisations of process equivalences by modal fixpoint formulas give a uniform treatment of
the descriptive complexity of such equivalence relations. However, they do not (yet) provide an
algorithmic treatment. The aim of this section is to do so. To this end, we explain how to do model
checking for $\HOHDMU$. In fact, much less suffices already. Remember that the input to a model
checking procedure is a pair consisting of --- typically --- an LTS and a formula.
Higher-dimensionality of the underlying logic means that the input is a pair consisting of a tuple of
LTS on one side and a formula on the other.  Now any algorithm that does model checking for a pair of
LTS and any formula $\Phi_{\mathcal{R}}$ given in the previous section is in fact an algorithm that
decides the process equivalence $\mathcal{R}$.  Thus, for these purposes it suffices to explain how to
do model checking for any fragment that encompasses the formulas given there.

Here we restrict our attention to the fragment $\hohdmu{1}{2}$. This captures all process equivalences
considered here apart from possible-futures equivalence, because all their characteristic formulas are
naturally of dimension 2 --- they describe a binary relation --- and are of order 1. The extension to
higher dimensionality is straight-forward. The extension to higher orders is also possible but not done 
here for ease of presentation.

\subsection{A Symbolic Model-Checking Algorithm}

We give a model checking algorithm for $\hohdmu{1}{2}$ that can be seen as a suitable extension of the
usual fixpoint iteration algorithm for the modal $\mu$-calculus. It merges the ideas used in model
checking for the higher-dimension $\mu$-calculus \cite{LL-FICS12} and for higher-order fixpoint logic 
\cite{als-mchfl07,al-lpar07}. 

Let $\Phi$ be a well-typed formula of $\hohdmu{1}{2}$. Then each of its subformulas has a type of the
form $\iPR{2}^{v_1} \to \ldots \to \iPR{2}^{v_m} \to \iPR{2}$ for some $m \ge 0$. Algorithm~\ref{alg:mc} 
takes as input two LTS $\Transsys_i = (\Procs_i,\Act,\rightarrow_i)$ for $i \in \{1,2\}$ and an 
$\hohdmu{1}{2}$ formula $\Phi$, and returns the set of all pairs of processes from these two LTS that 
satisfy $\Phi$. 
%We assume that all subformulas are normalised to be of the form $\Psi\ \Psi_1\ \ldots\ \Psi_m$ 
%such that $\Psi$ has the type stated above and all $\Psi_i$ have type $\iPR{2}$. 
Model checking is done
by simply computing the semantics of each such subformula on the two underlying LTS. 

\algblockdefx{Cases}{EndCases}%
   [1]{\textbf{case} #1 \textbf{of}}%
   {\textbf{end case}}
\algcblockx[Cases]{Cases}{Case}{EndCases}%
   [1]{#1:\enspace}%
   {\textbf{end case}}

\begin{algorithm}[t]
\caption{Model Checking $\hohdmu{1}{2}$}
\label{alg:mc}
\begin{algorithmic}[1]
\Procedure{\textsc{MC}}{$\Phi,\rho$} \Comment{assume $\Transsys_i = (\Procs_i,\Act,\rightarrow_i)$ to be fixed for $i=1,2$}
\Cases{$\Phi$}
\Case{$\top$} \hspace*{9.7mm}\textbf{return} $\Procs_1 \times \Procs_2$
\Case{$x$} \hspace*{11mm}\textbf{return} $\rho(x)$ \Comment{some variable of type $\iPR{2}^{v_1} \to \ldots \to \iPR{2}^{v_m} \to \iPR{2}$}
\Case{$\neg\Psi$} \hspace*{7.3mm}\textbf{return} $(\Procs_1 \times \Procs_2)\ \setminus\ $ \textsc{MC}($\Psi,\rho$)
\Case{$\Psi_1 \wedge \Psi_2$} \textbf{return} \textsc{MC}($\Psi_1,\rho$) $\cap$ \textsc{MC}($\Psi_2,\rho$)
\Case{$\may{a}{1}\Psi$} \hspace*{3.5mm}\textbf{return} $\{ (P_1,P_2) \mid \exists P' \in \Procs_1$ s.t.\ $P_1 \tr{a}_1 P'$ and $(P',P_2) \in$ \textsc{MC}($\Psi,\rho$) $\}$
\Case{$\may{a}{2}\Psi$} \hspace*{3.5mm}\textbf{return} $\{ (P_1,P_2) \mid \exists P' \in \Procs_2$ s.t.\ $P_2 \tr{a}_2 P'$ and $(P_1,P') \in$ \textsc{MC}($\Psi,\rho$) $\}$
\Case{$\lambda x_1,\ldots,x_m:\iPR{2}^{v_1} \to\ldots\to\iPR{2}^{v_m} \to \iPR{2}$}
  \ForAll{$(T_1,\ldots,T_m) \in (2^{\Procs_1 \times \Procs_2})^m$}
    \State $F(T_1,\ldots,T_m) \gets$ \textsc{MC}($\Psi,\eta[x_1 \mapsto T_1,\ldots,x_m \mapsto T_m]$)
  \EndFor
  \State \textbf{return} $F$
% % \Case{$\mu x:\iPR{2}.\Psi$}
% %   \State $T \gets \emptyset$
% %   \Repeat
% %     \State $T' \gets T$
% %     \State $T \gets$ \textsc{FPIter}($\Psi,\rho[x \mapsto T']$)
% %   \Until{$T = T'$}
% %   \State \textbf{return} $T$
\Case{$\Psi\ \Psi_1\ \ldots\ \Psi_m$}
  \State \textbf{return} \textsc{MC}($\Psi,\rho$)(\textsc{MC}($\Psi_1,\rho$),\ldots,\textsc{MC}($\Psi_m,\rho$)) 
\Case{$\mu x:\iPR{2}^{v_1} \to \ldots \to \iPR{2}^{v_m} \to \iPR{2}.\Psi$}
  \ForAll{$(T_1,\ldots,T_m) \in (2^{\Procs_1 \times \Procs_2})^m$}
    \State $F(T_1,\ldots,T_m) \gets \emptyset$
  \EndFor
  \Repeat
    \State $F' \gets F$
    \ForAll{$(T_1,\ldots,T_m) \in (2^{\Procs_1 \times \Procs_2})^m$}
      \State $F(T_1,\ldots,T_m) \gets$ \textsc{MC}($\Psi,\rho[x \mapsto F']$)
    \EndFor
  \Until{$F = F'$}
  \State \textbf{return} $F$
\EndCases
\EndProcedure
\end{algorithmic}
\end{algorithm}

The difference to model checking the modal $\mu$-calculus is the handling of higher-order subformulas.
Note that the semantics of a function of type $\iPR{2}^{v_1} \to \ldots \to \iPR{2}^{v_m} \to \iPR{2}$
over a pair of LTS with $n_1$, respectively $n_2$ many processes can be represented as a table with
$(2^{n_1 \cdot n_2})^m$ many entries --- one for each possible combination of argument values to this
function. Algorithm \textsc{MC} is designed to compute such a table for the corresponding subformulas.

\begin{theorem}
Let $\Phi$ be a closed $\hohdmu{1}{2}$ formula of size $k$, and $\Transsys_1,\Transsys_2$ be two finite LTS,
each of size $n$ at most. The call of \textsc{MC}($\Phi,[]$) correctly computes 
$\sem{\judg{\emptyset}{\Phi: \iPR{2}}}$ with respect to $\Transsys_1,\Transsys_2$ in time 
$\mathcal{O}(n^2 \cdot 2^{n^2 k^2})$.
\end{theorem}

\begin{proof}
(Sketch) Correctness is established through a straight-forward induction on the structure of $\Phi$. Note
that the theorem is too weak to be used as an inductive invariant. Instead, one can easily
prove the following stronger assertion: for any provable statement $\judg{\Gamma}{\Psi:\tau}$ and any 
interpretation $\eta$, \textsc{MC}($\Psi,\eta$) computes $\sem{\judg{\Gamma}{\Psi:\tau}}_{\eta}$. For
most cases this follows immediately from the definition of the semantics and the induction hypothesis.
For fixpoint formulas it also uses the well-known characterisation of least fixpoints by their chain
of approximants. Note that the underlying power lattice is finite, even for higher-order types. Thus, fixpoint 
iteration from below --- as done in algorithm \textsc{MC} --- converges to the least fixpoint of the
corresponding function in a finite number of steps.

The upper bound on the worst-case running time is established as follows. Note that $k$ is an upper bound
on the arity of each subformulas first-order type, i.e.\ in $\iPR{2}^{v_1} \to \ldots \to \iPR{2}^{v_m} \to \iPR{2}$
we have $m \le k$. Clearly, the running time for each case-clause is dominated by the one for fixpoint
formulas which --- disregarding recursive calls --- can be done in time $\mathcal{O}(n^2 \cdot 2^{n^2 k})$. 
Note that it needs to fill a table with $2^{n^2 k}$ many entries using fixpoint iteration. Each table entry
can change at most $n^2$ many times due to monotonicity. Furthermore, note that it is not the case that the
semantics of each subformula is only computed once. Because of nested fixpoint formulas, we obtain an
additional exponent which is bounded by the number of fixpoint formulas, i.e.\ also bounded by $k$, resulting
in an upper bound of $\mathcal{O}(n^2 \cdot 2^{n^2 k^2})$.
\end{proof}

This establishes exponential-time upper bounds for all the process equivalence relations which can be
defined in $\hohdmu{1}{2}$. 

\begin{corollary}
Trace, completed trace, failure, failure trace, readiness and ready trace equivalence can be
checked in time $2^{\mathcal{O}(n^2)}$.
\end{corollary}

It is easily checked that for $\hohdmu{0}{2}$ formulas, algorithm \textsc{MC} runs in time $\mathcal{O}((kn^2)^k)$.
By instantiation we obtain polynomial-time algorithms for further process equivalences. 

\begin{corollary}
Completed, ready, 2-nested, bi- and simulation equivalence can be checked in polynomial time.
\end{corollary}

We point out that algorithm \textsc{MC} can be made to work symbolically on BDDs just like 
the algorithm for the $\mu$-calculus can \cite{IC::BurchCMDH1992}. A function is then represented as a table
of BDDs. Furthermore, it can straight-forwardly be extended to higher orders which increases the complexity
by one exponential per order. As a result, we obtain the following.

\begin{proposition}
Possible-futures equivalence can be checked in doubly exponential time.
\end{proposition}

\subsection{Need-Driven Function Evaluation}

Algorithm \textsc{MC} computes values for functions in a very na\"{\i}ve and brute-force way: it tabulates
all possible arguments to the function and computes all their values. This results in far too many 
value computations than are needed in order to compute $\sem{\judg{\emptyset}{\Phi:\iPR{2}}}$ for any
closed formula $\Phi$. Consider for example $(\lambda X^+:\iPR{2}. \must{a}{2}X)\ \bot$. Its semantics is
the set of all pairs $(P,Q)$ such that $Q$ has no $a$-successors. However, algorithm \textsc{MC} would compute the 
set of all pairs $(P,Q)$ such that all $a$-successors of $Q$ belong to the second components of any set 
of pairs $(P,Q')$. 

Need-driven function evaluation avoids these unnecessary computations. For formulas without fixpoint quantifiers
it could easily be realised by evaluating arguments first, and then passing these values to the computation
of the function, comparable to lazy evaluation in functional programming. Need-driven function evaluation in
the presence of fixpoint quantifiers is more complicated, though \cite{Jorgensen94-0:confs}. For recursively defined functions it is not
sufficient to simply compute their value on a given argument using fixpoint iteration for instance, but the
computation of the value on some argument may need the value on some other argument. Need-driven function
evaluation intertwines the computation of these values with the exploration of the function's domain \cite{al-lpar07}. 
The following example shows the optimising potential of this technique.

\begin{figure}
\begin{center}
\raisebox{15mm}{\scalebox{0.7}{\begin{tikzpicture}[thick, node distance=1.7cm, initial text={}, every state/.style={minimum size=4mm},]
  \node[state,initial] (q0)               {0};
  \node[state]         (q1) [right of=q0] {1};

  \path[->] (q0) edge [bend left]  node [above] {$b$} (q1)
            (q1) edge [bend left]  node [below] {$b$} (q0)
                 edge [loop right] node [right] {$a$} ();

  \node[state,initial] (p0) [right of=q1,node distance=2.5cm] {2};
  \node[state]         (p1) [right of=p0]                     {3};
  \node[state]         (p2) [below of=p1]                     {4};

  \path[->] (p0) edge [bend left]  node [above]      {$b$} (p1)
            (p1) edge [bend left]  node [right]      {$a$} (p2)
                 edge [bend left]  node [below]      {$b$} (p0)
            (p2) edge [bend left]  node [below left] {$b$} (p0)
                 edge [bend left]  node [left]       {$a$} (p1);
\end{tikzpicture}}}
\hspace*{2cm}
\scalebox{0.7}{\begin{tikzpicture}[dashed, node distance=3cm, initial text={}, every state/.style={minimum size=4mm},]
  \node (A)               {$\begin{array}{c}S_1 \times S_2 \\ \emptyset \end{array}$};
  \node (B) [right of=A]  {$\begin{array}{c}\{1\} \times S_2 \\ S_1 \times \{2\}\end{array}$};
  \node (C) [below of=B]  {$\begin{array}{c}\{0\} \times S_2 \\ S_1 \times \{3,4\}\end{array}$};
  \node (D) [below of=A]  {$\begin{array}{c}\emptyset \\ S_1 \times S_2\end{array}$};

  \path[->] (A) edge              node [above] {\first} (B)
                edge [loop above] node [above] {\second} ()
            (B) edge [loop above] node [above] {\first} ()
                edge [bend left]  node [right] {\second} (C)
            (C) edge              node [below] {\first} (D)
                edge [bend left]  node [left]  {\second} (B)
            (D) edge [loop above] node [above] {\first,\second} ();
\end{tikzpicture}}
\vskip4mm

\begin{math}
\begin{array}{c@{\enspace}|c@{\enspace}c@{\enspace}c@{\enspace}c}
X & S_1 {\times} S_2 & \{1\}{\times}S_2 & \{0\}{\times}S_2 & \emptyset \\
Y & \emptyset & S_1{\times}\{2\} & S_1{\times}\{3,4\} & S_1{\times}S_2 \\
\hline
\mathcal{F}^0 & \emptyset & \emptyset & \emptyset & \emptyset \\
\mathcal{F}^1 & \emptyset & \{(1,2)\} & \{(0,3),(0,4)\} & \emptyset \\
\mathcal{F}^2 & \{(1,2)\} & \{(1,2),(0,3),(0,4)\} & \{(1,2),(0,3),(0,4)\} & \emptyset \\
\mathcal{F}^3 & \{(1,2),(0,3),(0,4)\} & \{(1,2),(0,3),(0,4)\} & \{(1,2),(0,3),(0,4)\} & \emptyset \\
\mathcal{F}^4 & \{(1,2),(0,3),(0,4)\} & \{(1,2),(0,3),(0,4)\} & \{(1,2),(0,3),(0,4)\} & \emptyset
\end{array}
\end{math}
\end{center}

\caption{Example of need-driven function evaluation for trace equivalence checking.}
\label{fig:example}
\end{figure}

\begin{example}
Consider the two LTS presented in Fig.~\ref{fig:example}. Let $S_1 \eqdef \{0,1\}$ and 
$S_2 \eqdef \{2,3,4\}$ be their state spaces. We will show how need-driven function evaluation works on algorithm \textsc{MC}, 
these two LTS and the formula that tests for trace equivalence over $\Act = \{a,b\}$, namely
\begin{displaymath}
\Phi_{\mathsf{t}} \enspace = \enspace \big(\mu F(X,Y).\ (X\wedge Y)\ \vee\ (F\ \may{a}{1}X \ \must{a}{2}Y) \ \vee \ 
(F\ \may{b}{1}X\ \must{b}{2}Y)\big)\enspace \top \enspace \bot\enspace .
\end{displaymath}
Note that it should be true on a pair $(P,Q)$ of processes iff $P$ has a trace that $Q$ does not. 

$\Phi_{\mathsf{t}}$ defines a function $\mathcal{F}$ via least-fixpoint recursion. It takes two arguments $X$ and $Y$ and
returns the union of their intersection with the value of $\mathcal{F}$ applied to two other sets of arguments,
defined by $\may{a}{1}X$ and $\must{a}{2}Y$ in one case and equally with $b$ in the other. Moreover, we are interested
in the value of $\mathcal{F}$ on the argument pair $(S_1\times S_2,\emptyset)$. 

Need-driven function evaluation builds the table for $\mathcal{F}$ via fixpoint iteration, i.e.\ by building its 
approximants $\mathcal{F}^0$, $\mathcal{F}^1$, \ldots with $\mathcal{F}^0(X,Y) = \emptyset$ for any 
$X,Y \subseteq S_1\times S_2$, starting with the argument on which we need the function's value. Since
$\mathcal{F}$ is recursively defined, the value on this argument may need the value on other arguments. 
Fig.~\ref{fig:example} shows the part of the dependency graph that is reachable from this initial argument, where
an arrow $(X,Y) \stackrel{a}{\dasharrow} (X',Y')$ states that the computation of the value on $(X,Y)$ triggers the
first recursive call on $(X',Y')$. Similarly, an arrow $\stackrel{b}{\dasharrow}$ shows the dependency via the second
recursive call.

Finally, Fig.~\ref{fig:example} shows the table of values computed by fixpoint iteration restricted to those 
arguments that occur in the dependency graph, i.e.\ the part of the function's domain which is necessary to iterate
on until stability in order to determine the fixpoint's value on the initial argument. The optimising potential of
need-driven function evaluation is justified by the table's width: note that the na\"{\i}ve version of algorithm
\textsc{MC} would fill that table for all possible arguments of which there are $(2^{2\cdot 3})^2 = 4096$ while it
suffices to reach stability on these $4$ arguments alone.
\end{example}

\subsection{Partial Evaluation}

The example above shows another potential for optimisation. Remember that the formals defining process equivalences
do not depend on the actual LTS on which they are being evaluated. Thus, we can devise a simpler algorithm for trace
equivalence for instance by analysing the behaviour of \textsc{MC} on an arbitrary pair of LTS and the fixed formula
$\Phi_{\mathsf{t}}$. We note that the filling of the table values follows a simple scheme: the value in row $i$ at
position $(X,Y)$ is the union of three values, namely the one in row $1$ of this position and the values in row
$i-1$ of the two successors of $(X,Y)$ in the dependency graph. This leads to the simple Algorithm~\ref{alg:treq} for
trace equivalence checking.
  
\begin{algorithm}[t]
\caption{Trace Equivalence Checking}
\label{alg:treq}
\begin{algorithmic}[1]
\Procedure{\textsc{TrEq}}{$\Transsys_1,\Transsys_2$} \Comment{let $\Transsys_i = (\Procs_i,\Act,\rightarrow_i)$}
\State $X_0 \gets \Procs_1\times\Procs_2$
\State $Y_0 \gets \emptyset$
\State $\mathcal{W} = \{ (X_0,Y_0) \}$ \Comment{work list}
\State $\mathcal{D} = \emptyset$ \Comment{domain of the dependency graph}
\While{$\mathcal{W} \ne \emptyset$} \Comment{build dependency graph}
  \State remove some $(X,Y)$ from $\mathcal{W}$
  \ForAll{$a \in \Act$}
    \State $(X',Y') \gets (\may{a}{1}X,\must{a}{2}Y)$ 
    \State $d_a(X,Y) \gets (X',Y')$ \Comment{record arrows in dependency graph}
    \State $\mathcal{D}�\gets \mathcal{D} \cup \{(X,Y)\}$
    \If{$(X',Y') \not\in \mathcal{D}$}
      \State $\mathcal{W} \gets \mathcal{W} \cup \{(X',Y')\}$
    \EndIf
  \EndFor
\EndWhile
\ForAll{$(X,Y) \in \mathcal{D}$}
  \State $I(X,Y) \gets X \cap Y$
  \State $F(X,Y) \gets \emptyset$
\EndFor
\Repeat
  \ForAll{$(X,Y) \in \mathcal{D}$}
    \State $F(X,Y) \gets I(X,Y) \cup \bigcup_{a \in \Act} F(d_a(X,Y))$
  \EndFor
\Until{$F$ does not change anymore}
\State \textbf{return} $F(X_0,Y_0)$
\EndProcedure
\end{algorithmic}
\end{algorithm}

%%% Local Variables: 
%%% mode: latex
%%% TeX-master: "main"
%%% End: 

\section{Conclusion and Further Work}
\label{sec:concl}
We have presented a highly expressive modal fixpoint logic which can define many process equivalence
relations. We have presented a model checking algorithm which can be instantiated in order to yield
decision procedures for the relations on finite systems. This re-establishes already known decidability 
results \cite{ncstrl.albany_cs//SUNYA-CS-96-03}. Its main contribution, though, is the --- to the best
of our knowledge --- first framework that provides a generic and uniform algorithmic approach to process 
equivalence checking via defining formulas. In particular, it allows technology from the well-developed field of model checking 
to be transferred to process equivalence checking. 

There is a lot of potential further work into this direction. The exponential-time bound for the
trace-like equivalences is not optimal since they are generally PSPACE-complete
\cite{ncstrl.albany_cs//SUNYA-CS-96-03}.  It remains to be seen whether the formulas defining them have
a particular structure that would allow a PSPACE model checking algorithm for instance. This would make a
real improvement since model checking $\hohdmu{1}{2}$ is EXPTIME-hard in general which follows from such
a bound for the first-order fragment of HFL \cite{als-mchfl07}. Also, it
remains to be seen whether or not possible-futures equivalence can be defined $\hohdmu{1}{2}$.

We leave the exact formulation of a model checking procedure for the entire logic $\HOHDMU$ for future work.
Such an algorithm may be interesting for other fields as well, not just process equivalence checking.

There are more equivalence relations which we have not considered here for lack of space, e.g.\ possible-worlds
equivalence, tree equivalence, 2-bounded trace bisimulation, etc. We believe that creating defining formulas
for them in $\HOHDMU$ is of no particular difficulty. 

We intend to also investigate the practicability of this approach. To this end, we aim to extend an existing
prototypical implementation of a symbolic model checking tool for the higher-dimension $\mu$-calculus to 
$\hohdmu{1}{2}$, and possible $\HOHDMU$ in general. We believe that using need-driven function evaluation and
partial evaluation techniques will have a major influence on the applicability of the algorithms obtained by 
instantiating the generic model checking procedure with a fixed formula.

%%% Local Variables: 
%%% mode: latex
%%% TeX-master: "main"
%%% End: 

\bibliographystyle{eptcs}
\bibliography{biblio}

\appendix
\section{Definitions of Process Equivalences}
\label{sec:app}
\parindent0pt

\textbf{Finite Trace Equivalence.}
Let $T(P) := \{t \mid \exists Q. P \tr{t} Q\}$ be the set of all finite traces of $P$. Two processes $P$ and $Q$ are 
\emph{finite trace equivalent}, $P \sim_{\mathsf{t}} Q$, if $T(P) = T(Q)$.

\textbf{Completed Trace Equivalence.}
A sequence $t \in \Act^*$ of a process $P$ is a \emph{completed trace} if there is a $Q$ s.t. $P \tr{t} Q$ and 
$I(Q) = \emptyset$. Let $CT(P)$ be the set of all completed traces of $P$. Two processes $P$ and $Q$ are \emph{completed trace equivalent},
$P \sim_{\mathsf{ct}} Q$, if $T(P) = T(Q)$ and $CT(P) = CT(Q)$.

\textbf{Failures Equivalence.}
A pair $\langle t, A \rangle$ is a \emph{failure pair} of $P$ if there is a process $Q$ s.t. $P \tr{t} Q$ and 
$I(Q) \cap A = \emptyset$. Let $F(P)$ denote the set of all failure pairs of $P$. Two processes $P$ and $Q$ are \emph{failures equivalent}, 
$P \sim_{\mathsf{f}} Q$, if $F(P) = F(Q)$.

\textbf{Failure Trace Equivalence.}  A \emph{failure trace} is a $u \in (\Act \cup
2^{\Act})^*$. We extend the reachability relation of processes to failure traces by including 
$P \tr{\varepsilon}_{\mathsf{ft}} P$ for any $P$ and the triples $P \tr{A}_{\mathsf{ft}} Q$
whenever $I(P) \cap A = \emptyset$, and then closing it off under compositions: if 
$P\tr{u}_{\mathsf{ft}} R$ and $R \tr{u'}_{\mathsf{ft}} Q$ then $P \tr{uu'}_{\mathsf{ft}} Q$.  
Let $FT(P) := \{ u \mid \exists Q. P \tr{u}_{\mathsf{ft}} Q \}$ 
be the set of all failure traces of $P$. Two
processes $P$ and $Q$ are \emph{failure trace equivalent}, $P \sim_{\mathsf{ft}} Q$, if
$FT(P) = FT(Q)$.

\textbf{Readiness Equivalence.}
A pair $\langle t, A \rangle$ is a \emph{ready pair} of $P$ if there is a process $Q$ s.t. $P \tr{t} Q$ and $A = I(Q)$. Let 
$R(P)$ denote the set of all ready pairs of $P$. Two processes $P$ and $Q$ are \emph{ready equivalent}, $P \sim_{\mathsf{r}} Q$, if $R(P) = R(Q)$.

\textbf{Ready Trace Equivalence.}
A \emph{ready trace} is a $u \in (\Act \cup 2^{\Act})^*$. We extend the reachability relation of processes to ready traces by
including $P \tr{\varepsilon}_{\mathsf{rt}} P$, and $P \tr{A}_{\mathsf{rt}} Q$ whenever $I(P) = A$, and closing it off under
compositions as in the case of failure trace equivalence. Let $RT(P) := \{u \mid \exists Q. P \tr{u}_{\mathsf{rt}} Q \}$ 
be the set of all ready traces of $P$. Two processes $P$ and $Q$ are \emph{ready trace equivalent}, $P \sim_{\mathsf{rt}} Q$, if 
$RT(P) = RT(Q)$.

\textbf{Possible-Futures Equivalence.}
A pair $\langle t, L \rangle$ is a \emph{possible future} of $P$ if there is a process $Q$ s.t. $P \tr{t} Q$ and $L = T(Q)$. 
Let $PF(P)$ be the set of all possible futures of $P$. Two processes $P$ and $Q$ are \emph{possible-futures equivalent}, $P \sim_{\mathsf{pf}} Q$, 
if $PF(P) = PF(Q)$.

\textbf{Simulation Equivalence.}
A binary relation $\mathcal{R}$ is a \emph{simulation} on processes if it satisfies for any $a \in \Act$:
if $(P,Q) \in \mathcal{R}$ and $P \tr{a} P'$, then $\exists Q'. Q \tr{a} Q'$ and 
$(P',Q') \in \mathcal{R}$. 
$P$ and $Q$ are \emph{similar}, $P \sim_{\mathsf{s}} Q$, if there are simulations $\mathcal{R}$ and $\mathcal{R'}$ s.t. $(P,Q) \in \mathcal{R}$ and 
$(Q,P) \in \mathcal{R'}$.

\textbf{Completed Simulation Equivalence.}
A binary relation $\mathcal{R}$ is a \emph{completed simulation} on processes if it satisfies for any $a \in \Act$:
if $(P,Q) \in \mathcal{R}$ and $P \tr{a} P'$, then $\exists Q'. Q \tr{a} Q'$ and 
$(P',Q') \in \mathcal{R}$.
And if $(P,Q) \in \mathcal{R}$ then $I(P) = \emptyset \Leftrightarrow I(Q) = \emptyset$.
Two processes $P$ and $Q$ are \emph{completed simulation equivalent}, $P \sim_{\mathsf{cs}} Q$, if there are completed simulations $\mathcal{R}$ and 
$\mathcal{R'}$ s.t. $(P,Q) \in \mathcal{R}$ and $(Q,P) \in \mathcal{R'}$.

\textbf{Ready Simulation Equivalence.}
A binary relation $\mathcal{R}$ is a \emph{ready simulation} on processes if it satisfies for any $a \in \Act$:
if $(P,Q) \in \mathcal{R}$ and $P \tr{a} P'$, then $\exists Q'. Q \tr{a} Q'$ and 
$(P',Q') \in \mathcal{R}$.
And if $(P,Q) \in \mathcal{R}$ then $I(P) = I(Q)$.
Two processes $P$ and $Q$ are \emph{ready simulation equivalent}, $P \sim_{\mathsf{rs}} Q$, if there are ready simulations $\mathcal{R}$ and
$\mathcal{R'}$ s.t. $(P,Q) \in \mathcal{R}$ and $(Q,P) \in \mathcal{R'}$.

\textbf{2-Nested Simulation Equivalence.} 
A binary relation $\mathcal{R}$ is a \emph{2-nested simulation} on processes if it satisfies for any $a \in \Act$:
if $(P,Q) \in \mathcal{R}$ and $P \tr{a} P'$, then $\exists Q'. Q \tr{a} Q'$ and 
$(P',Q') \in \mathcal{R}$.
And if $(P,Q) \in \mathcal{R}$ then $Q \sim_{\mathsf{s}} P$.
Two processes $P$ and $Q$ are \emph{2-nested simulation equivalent}, $P \sim_{\mathsf{2s}} Q$, if there are 2-nested simulations $\mathcal{R}$ and
$\mathcal{R'}$ s.t. $(P,Q) \in \mathcal{R}$ and $(Q,P) \in \mathcal{R'}$.

\textbf{Bisimulation.}
A binary relation $\mathcal{R}$ is a \emph{bisimulation} on processes if it satisfies for any $a \in \Act$:
if $(P,Q) \in \mathcal{R}$ and $P \tr{a} P'$, then $\exists Q'. Q \tr{a} Q'$ and 
$(P',Q') \in \mathcal{R}$. 
And if $(P,Q) \in \mathcal{R}$ and $Q \tr{a} Q'$, then $\exists P'. P \tr{a} P'$ and
$(P',Q') \in \mathcal{R}$.
Two processes $P$ and $Q$ are \emph{bisimilar}, $P \sim_{\mathsf{b}} Q$, if there is a bisimulation $\mathcal{R}$ s.t. $(P,Q) \in \mathcal{R}$.

%%% Local Variables: 
%%% mode: latex
%%% TeX-master: "main"
%%% End: 

\end{document}